\documentclass[
notitlepage
,floatfix,
aps,
pra,
reprint,
superscriptaddress,
]{revtex4-1}
\usepackage[utf8]{inputenc}

\usepackage[caption=false]{subfig}
\usepackage{graphicx} 
\usepackage{epstopdf}
\usepackage{array}
\usepackage{verbatim}
\usepackage{amsmath,amsfonts,amssymb,amscd,mathtools}
\usepackage{amsthm}
\usepackage{tabularx}
\usepackage{stmaryrd}
\usepackage{enumerate}
\usepackage[ruled]{algorithm2e}
\usepackage{wasysym}
\usepackage{braket}
\usepackage{mathrsfs}
\usepackage{microtype}
\usepackage{hyperref}
\usepackage[dvipsnames]{xcolor}
\hypersetup{
    bookmarksnumbered=true, 
    unicode=false, 
    pdfstartview={FitH}, 
    pdftitle={}, 
    pdfauthor={}, 
    pdfsubject={}, 
    pdfcreator={}, 
    pdfproducer={}, 
    pdfkeywords={}, 
    pdfnewwindow=true, 
    colorlinks=true, 
    linkcolor=NavyBlue, 
    citecolor=NavyBlue, 
    filecolor=NavyBlue, 
    urlcolor=NavyBlue 
}

\usepackage{newtxtext,newtxmath}

\theoremstyle{plain}
\newtheorem{thm}{Theorem}

\theoremstyle{definition}

\newcommand{\eq}[1]{(\hyperref[eq:#1]{\ref*{eq:#1}})}

\renewcommand{\sec}[1]{\hyperref[sec:#1]{Section~\ref*{sec:#1}}}
\newcommand{\thrm}[1]{\hyperref[thrm:#1]{Theorem~\ref*{thrm:#1}}}
\newcommand{\lemm}[1]{\hyperref[lemm:#1]{Lemma~\ref*{lemm:#1}}}
\newcommand{\prop}[1]{\hyperref[prop:#1]{Proposition~\ref*{prop:#1}}}
\newcommand{\corr}[1]{\hyperref[corr:#1]{Corollary~\ref*{corr:#1}}}
\newcommand{\fig}[1]{\hyperref[fig:#1]{~\ref*{fig:#1}}}
\newcommand{\deff}[1]{\hyperref[deff:#1]{~\ref*{deff:#1}}}

\newcommand{\mA}{\mathcal{A}}

\newcommand{\mE}{\mathcal{E}}

\newcommand{\mU}{\mathcal{U}}
\newcommand{\mT}{\mathcal{T}}
\newcommand{\mD}{\mathcal{D}}
\newcommand{\mI}{\mathcal{I}}
\newcommand{\mF}{\mathcal{F}}

\newcommand{\mH}{\mathcal{H}}

\newcommand{\mO}{\mathcal{O}}
\newcommand{\mB}{\mathcal{B}}
\newcommand{\mK}{\mathcal{K}}
\newcommand{\mP}{\mathcal{P}}

\newcommand{\mS}{\mathcal{S}}
\newcommand{\mX}{\mathcal{X}}
\newcommand{\mY}{\mathcal{Y}}
\newcommand{\mZ}{\mathcal{Z}}

\newcommand{\mbI}{\mathbb{I}}
\newcommand{\mbH}{\mathbb{H}}

\newcommand{\mbR}{\mathbb{R}}

\newcommand{\mV}{\mathcal{V}}

\renewcommand{\*}{\textup{*}}
\DeclareMathOperator{\cone}{cone}

\DeclareMathAlphabet{\matheu}{U}{eus}{m}{n}

\DeclareMathOperator{\Tr}{Tr}
\DeclareMathOperator{\id}{id}

\DeclareMathOperator{\sgn}{sgn}
\DeclareMathOperator{\tot}{tot}

\newcommand{\ketbra}[2]{|{#1}\rangle\!\langle{#2}|}

\newcommand{\ba}{\begin{eqnarray}}
\newcommand{\ea}{\end{eqnarray}}
\newcommand{\bann}{\begin{eqnarray*}}
\newcommand{\eann}{\end{eqnarray*}}
\newcommand{\bal}{\begin{equation}\begin{aligned}}
\newcommand{\eal}{\end{aligned}\end{equation}}
\newcommand{\dm}[1]{\ketbra{#1}{#1}}

\newcolumntype{L}[1]{>{\raggedright}p{#1}}
\newcolumntype{C}[1]{>{\centering}p{#1}}
\newcolumntype{R}[1]{>{\raggedleft}p{#1}}
\newcolumntype{D}{>{\centering\arraybackslash}X}

\renewcommand{\*}{\textup{*}}

\newcommand{\sbar}{\;\rule{0pt}{9.5pt}\right|\;}

\newcommand{\lset}{\left\{\left.}
\newcommand{\rset}{\right\}}

\begin{document}

\title{Optimal resource cost for error mitigation}

\author{Ryuji Takagi}
\email{ryuji.takagi@ntu.edu.sg}
\affiliation{Nanyang Quantum Hub, School of Physical and Mathematical Sciences, Nanyang Technological University, 637371, Singapore}
\affiliation{Center for Theoretical Physics and Department of Physics, Massachusetts Institute of Technology, Cambridge, MA 02139, USA}
\affiliation{QunaSys Inc., Aqua Hakusan Building 9F, 1-13-7 Hakusan, Bunkyo, Tokyo 113-0001, Japan}

\begin{abstract}
One of the central problems for near-term quantum devices is to understand their ultimate potential and limitations. We address this problem in terms of quantum error mitigation by introducing a framework taking into account the full expressibility of near-term devices, in which the optimal resource cost for the probabilistic error cancellation method can be formalized. We provide a general methodology for evaluating the optimal cost by connecting it to a resource-theoretic quantifier defined with respect to the noisy operations that devices can implement. We employ our methods to estimate the optimal cost in mitigating a general class of noise, where we obtain an achievable cost that has a generic advantage over previous evaluations, as well as a fundamental lower bound applicable to a broad class of noisy implementable operations. We improve our bounds for several noise models, where we give the exact optimal costs for the depolarizing and dephasing noise, precisely characterizing the overhead cost while offering an operational meaning to the resource measure in terms of error mitigation. Our result particularly implies that the heuristic approach presented by Temme \textit{et al}. [K. Temme, S. Bravyi, and J. M. Gambetta, Phys. Rev. Lett. 119, 180509 (2017)] is optimal even in our extended framework, putting fundamental limitations on the advantage provided by the extra degrees of freedom inherent in near-term devices for this noise model.  
\end{abstract}

\maketitle

\section{Introduction}

Recent technological developments push us toward the realization of quantum information processing in a fully controlled manner, and a near-term cornerstone is to make use of noisy intermediate-scale quantum (NISQ) devices, which focus on manipulating tens to hundreds of qubits without a full error correction~\cite{Preskill2018quantumcomputingin,McArdle2020quantum,arute_quantum_2019}.
However, whether these near-term devices can provide advantages in useful applications is still elusive.
In particular, rigorous theoretical analysis on the ultimate potential and limitations of NISQ devices has been largely missing.

In this work, we address this problem in terms of the capability to fight against noise. 
To deal with the critical noise effect without implementing an error-correcting code that is out of reach for the current technology, various error mitigation protocols have been proposed~\cite{Temme2017error, Li2017efficient, McClean2017subspace,Bonet2018lowcost, Endo2018practical,McArdle2019digital}. 
Among them, the probabilistic error cancellation method~\cite{Temme2017error,Endo2018practical,Sun2020continuous,Song2019exp,Zhang2020exp,Czarnik2020clifford, Strikis2020learning,Zlokapa2020deep} stands as a promising candidate, as it can construct an unbiased estimator that faithfully estimates the expectation value of an observable under a known error model. 
Probabilistic error cancellation is an application of a more general technique known as quasiprobability sampling~\cite{Troyer2005sign,Pashayan2015estimating,Howard2017application,Seddon2020quantifying,Buscemi2013twopoint,Buscemi2014twopoint}, and the relevant resource cost is the sampling overhead---the number of samples necessary to ensure a certain accuracy---characterized by how much negative portion the quasiprobability has. 
Thus, the capability of a given noisy device to mitigate errors with this method can be characterized by the minimum negativity in the quasiprobability, providing the optimal resource cost at which the probabilistic error cancellation can be run.

However, the ``optimal'' cost is not well-defined by itself.
Since quasiprobability decomposition expresses a given gate as a linear combination of noisy operations implementable on a noisy device, the optimal cost can highly depend on the choice of implementable operations. 
Many previous studies---not only on error mitigation but on NISQ algorithms in general---chose Pauli and Clifford gates as building blocks for the implementable operations~\cite{McArdle2020quantum}. 
As for the probabilistic error cancellation, Clifford operations combined with a target unitary gate followed by a noise channel were heuristically considered for depolarizing and amplitude damping noise~\cite{Temme2017error}, being extended to a complete set of Clifford-based operations applicable to a wider class of noise models~\cite{Endo2018practical}.
However, the necessity of considering the Pauli and Clifford gates has been rarely asked. 
It surely makes sense to give them special status in the fault-tolerant quantum computation with error-correcting codes where \emph{logical} operations take place in a code space, for which Clifford operations admit simple logical gate constructions~\cite{Steane1996error,Shor1995scheme,Fowler2012surface}.
On the other hand, for NISQ devices that do not implement error-correcting codes, there is no clear reason for the Clifford gates to be preferred over non-Clifford gates at the level of \textit{physical} operations on unencoded qubits.
In this sense, NISQ devices are endowed with extra degrees of freedom, and this flexibility should be fully exploited so that their potential and limitations can be properly gauged. 

Here, we introduce a framework that incorporates the full expressibility of noisy near-term devices that do not assume error correction. 
Our framework formalizes the optimal resource cost for the probabilistic error cancellation with respect to a continuous set of implementable operations, reflecting the flexibility of the NISQ devices. 
This consideration, however, also raises a demanding theoretical problem; since we need to take into account infinitely many implementable operations, obtaining the optimal resource cost can be extremely challenging.   
We address this problem by relating the optimal cost to a quantity studied in resource theories~\cite{Chitambar2019resource}.
A major pillar of resource theories is the quantification of resources, and previous studies indicated intimate relations between quasiprobability representation and resource quantifiers~\cite{Vidal2002computable,  Takagi2018convex, Albarelli2018resource, Tan2020negativity}. 
In particular, a resource quantifier known as \textit{robustness measure}~\cite{Vidal1999robustness}, has found several applications in the context of simulating quantum circuits~\cite{Howard2017application, Seddon2020quantifying} and quantum memory~\cite{Yuan2019memory}.
We consider our set of implementable operations as the accessible free resource and find that the optimal mitigation cost can be characterized by the robustness measure defined in our resource-theoretic framework, which can be evaluated by leveraging tools in general convex resource theories~\cite{Takagi2019operational,Takagi2019general}.

Then, we employ our method to obtain universal bounds for optimal error mitigation cost for a general noise model, showing that our framework provides a generic advantage over previous evaluations based on a discrete set of implementable operations, while placing a fundamental lower bound that must be observed by any device whose implementable operations are subsumed by the one introduced in our framework.   
We also study several specific noise models, for which we find that our bounds can be improved.
Notably, our methods provide \textit{exact} optimal costs for depolarizing and dephasing noise channels, precisely characterizing the error mitigation capability of noisy devices while offering an operational meaning to the robustness measure in the context of error mitigation.
Our result particularly indicates that the heuristic decomposition for the depolarizing channel given in Ref.~\cite{Temme2017error} is still optimal even in our extended framework, putting fundamental limitations on the enhancement enabled by a continuous set of implementable operations for this specific noise model.  

Our results not only provide a systematic way to evaluate the ultimate resource cost for error mitigation crucial for running useful algorithms in practice, but also present an application of ideas in resource theories to address concrete problems~\cite{Howard2017application,Halpern2020photoisomer,Takagi2018skew,Fang2019nogo,Takagi2020application,Zhou2020errorcorrection}, paving the way for a rigorous information-theoretical analysis of noisy near-term devices.

\section{Preliminaries}

A purpose of quantum computation, in particular for many variational algorithms designed for near-term devices, is to obtain an expectation value $\left<A\right>_{\rm ideal}=\Tr[A\rho_f]$, where $A$ is some observable and $\rho_f = \mU_L\circ\dots\circ\mU_1(\rho_i)$ is a final quantum state with input state $\rho_i$ fed into a quantum circuit composed of a sequence of unitary gates, $\{\mU_i\}_{i=1}^L$. (The curly letter refers to a unitary gate as a quantum channel, i.e., $\mU(\rho):=U\rho U^\dagger$.) 
However, since the application of quantum gates necessarily suffers from noise, the ideal gates $\{\mU\}_{i=1}^L$ are not directly implementable on noisy devices. 
Instead, one can consider a set of noisy operations, $\mI_\mE$, for some noise channel $\mE$ that is implementable on the device of interest. 
The idea of the probabilistic error cancellation method is to represent each quantum gate as a linear combination of the noisy implementable operations as $
\mU_i = \sum_\alpha \eta_{i,\alpha} \mO_\alpha, \mO_\alpha\in\mI_\mE$,
where $\eta_{i,\alpha}$ is a (not necessarily positive) real number.
Then, for each gate $\mU_i$ we sample an implementable operation $\mO_\alpha$ at probability $|\eta_{i,\alpha}|/\sum_\alpha|\eta_{i,\alpha}|$, prepare a state $\tilde\rho_f = \mO_{\alpha_L}\circ\dots\circ\mO_{\alpha_1}(\rho_i)$ where $\mO_{\alpha_i}$ is the implementable operation sampled for $\mU_i$, and measure the observable $A$. 
Then, letting $\gamma_i:=\sum_\alpha |\eta_{i,\alpha}|$, $\gamma_{\tot}:=\prod_{i=1}^L \gamma_i$, and $\sgn_{\tot}:=\prod_{i=1}^L\sgn(\eta_{i,\alpha_i})$, one can check that this realizes an unbiased estimator of $\left<A\right>_{\rm ideal}$ as $\left<A\right>_{\rm ideal}=\left<\gamma_{\tot}\sgn_{\tot}\mu(A)\right>_{\rm samp}$, where $\mu(A)$ is a random variable for the measurement outcome and $\left<\cdot\right>_{\rm samp}$ refers to the expectation value for the sampling average taken for the above procedure.

Although it gives the desired expectation value, canceling the noise comes with a cost: one needs to pay more sampling cost than would be needed to estimate the desired expectation value with a noiseless circuit within the same accuracy.
The Hoeffding's inequality~\cite{Hoeffding1963probability} ensures that a sufficient number of samples used for estimating the true expectation value with error $\delta$ at probability $1-\varepsilon$ is given by $(2\gamma_{\tot}^2/\delta^2)\ln(2/\varepsilon)$. 
Thus, having small $\gamma_{\tot}$ is crucial to suppress the sampling overhead and, since $\gamma_{\rm tot}$ grows exponentially with respect to the number of gates, the problem is reduced to finding a good linear decomposition of each ideal gate $\mU_i$ with respect to implementable operations $\mI_\mE$ with small $\gamma_i$. 

Clearly, the best linear decomposition depends on the choice of $\mI_\mE$, and it has been heuristically chosen on a case-by-case basis.
For instance, for the depolarizing noise model $\mD_{d,\epsilon}$ where $d$ is the dimension of the system and $\epsilon$ is the noise strength, $\mI_{\mD_{d,\epsilon}}=\{\mD_{d,\epsilon}\circ\mP\circ\mU\}$ where $\mP$ is a Pauli operator was considered, while for the single-qubit amplitude damping channel $\mA_\epsilon$, a set of implementable operations that works for the linear decomposition was found to be $\mI_{\mA_\epsilon}=\{\mA_\epsilon\circ\mU, \mA_\epsilon\circ\mZ^{1/2}\circ\mU,\, \mA_\epsilon\circ\mZ^{-1/2}\circ\mU,\, \mP_{\ket{0}}\}$ where $\mZ^{1/2}(\cdot)=Z^{1/2}\cdot {Z^{1/2}}^\dagger$ with $Z^{1/2}\coloneqq {\rm diag}(1,i)$ the phase gate and $\mP_{\ket{0}}$ the preparation of state $\ket{0}$~\cite{Temme2017error}. 
Later, this idea was extended to a Clifford-based universal basis set that works for any noise model with a sufficiently small noise strength~\cite{Endo2018practical}.

\section{Framework}

Although the above sets of operations can realize \textit{some} decomposition, there is no guarantee that these choices of $\mI_\mE$ achieve the smallest overhead $\gamma_i$ among other possible choices of the set of implementable operations. 
To assess the ultimate potential and limitations of the devices' capabilities having access to a continuous set of physical operations, we need to give more freedom to noisy devices with their programmability, i.e. the set of implementable operations under the noiseless condition.   
Motivated by this observation, we introduce the set of implementable operations as 
\bal
 \mI_\mE(d)=\lset  \mE\circ\Lambda \sbar \Lambda\in \tilde\mP(d)\rset
\label{eq:implementable set}
\eal
for a given noise channel $\mE\in\mT(d)$ where $\mT(d)$ is the set of completely positive trace preserving (CPTP) maps with input and output systems being $d$-dimensional quantum systems, and
$\tilde\mP(d)$ is the set of programmable operations defined as
\bal
 \tilde\mP(d)= \lset \sum_i p_i \mV_i\sbar \mV_i\in \mT_u(d)\cup \mS(d)\rset 
\eal
where $\{p_i\}$ is a probability distribution, $\mT_u(d)$ is the set of unitary channels on $d$-dimensional systems, and $\mS(d)=\{\mP_{\ket{\psi}}|\ket{\psi}\in\mH_d\}$ with $\mH_d$ being the $d$-dimensional Hilbert space is the set of state preparation channels.
The operations in $\mI_\mE(d)$ are implementable on the devices that can program any unitary gate and state preparation under the effect of noise $\mE$, which is reasonable for small $d$ such as $d=2$ (single qubit) and $d=4$ (two qubits), and even if a given device is not powerful enough to realize the above operations, our results serve as its ultimate bounds.
Note that although different operations may come with different noise channels in general, here we take a fixed noise channel after operations with the same size, which is a standard assumption in the quantitative analysis of noise effects~\cite{Temme2017error,Chamberland2016threshold,Fowler2012surface,Takagi2017error}; extending our formalism to accommodate different noise channels is left for future work. 

\section{Resource cost as robustness}

Our goal is to find a decomposition for a given ideal unitary gate $\mU\in\mT_u(d)$ with respect to the set of implementable operations \eqref{eq:implementable set} with the minimum absolute sum of the coefficients.
Namely, the optimal overhead constant $\gamma_{\rm opt}(\mU)$ is written as
\bal
 \gamma_{\rm opt}(\mU) = \min\lset\sum_\alpha|\eta_\alpha| \sbar\mU = \sum_\alpha \eta_\alpha\mO_\alpha,\ \mO_\alpha\in\mI_\mE(d)\rset
 \label{eq:gamma opt definition}
\eal
where we assume $\gamma_{\rm opt}(\mU)<\infty$. As noted in Ref.~\cite{Temme2017error},  Eq.~\eqref{eq:gamma opt definition} becomes a linear program when $\mI_\mE(d)$ is a discrete set.
However, because we aim to exploit the full expressibility of the device and consider the continuous set $\mI_\mE(d)$ given in \eqref{eq:implementable set}, Eq.~ \eqref{eq:gamma opt definition} is no longer simple linear programming; if one obtains some valid decomposition of $\mU$, it is hard to see whether there exists another decomposition that gives a smaller overhead constant. 

Here, we provide a general strategy to approach this problem. 
First, we introduce the following quantity
\bal
 R_{\mI_\mE}(\mU):= \min\lset s\geq 0 \sbar \frac{\mU + s \Xi}{1+s}\in \mI_\mE,\ \Xi\in\mI_\mE\rset.
 \label{eq:robustness primal}
\eal
This type of quantities defined for quantum states have been used to quantify the amount of quantum resources (e.g., entanglement) contained in the given state and is known as the (standard) robustness measure in resource theories.
In particular, the quantity in \eqref{eq:robustness primal} can be considered as the robustness measure defined for quantum operations in the context of the resource theory of channels~\cite{Liu2020operational,liu2019resource,Takagi2019general}, where $\mI_\mE$ serves as the set of free channels. 
Then, one can show that \eqref{eq:robustness primal} is equivalent to \eqref{eq:gamma opt definition} with the relation
\bal
 \gamma_{\rm opt}(\mU) = 2R_{\mI_\mE}(\mU) + 1,
\label{eq:gamma and robustness}
\eal
which we explain in Appendix~\ref{app:gamma and robustness}. (See also related arguments considered for other settings in Refs.~\cite{Howard2017application, Yuan2019memory, Regula2017convex}.)
Since $\mI_\mE(d)$ is a convex set, i.e. $\mO_1,\mO_2\in\mI_\mE(d)\Rightarrow p\mO_1+(1-p)\mO_2\in\mI_\mE(d),\ 0\leq p\leq 1$,
we can bring up ideas and tools developed in general convex resource theories~\cite{Takagi2019operational,Takagi2019general} and apply them to our resource measure in \eqref{eq:robustness primal}.
In particular, we obtain the following dual form of the robustness (details are given in Appendix~\ref{app:robustness dual}):
\bal
 R_{\mI_\mE}(\mU)= \max\lset \Tr[YJ_\mU] - 1 \sbar 
 0\leq \Tr[YJ_\Xi]\leq 1,\right.\\
 \left.\forall \Xi \in \mI_\mE,\ Y\in\mbH\rset
 \label{eq:robustness dual}
\eal
where 
$J_\Lambda:= \id\otimes \Lambda(d\cdot\Phi_d)$ with $\Phi_d = \frac{1}{d}\sum_{i,j=0}^{d-1}\ketbra{ii}{jj}$
is the Choi matrix of channel $\Lambda$ and $\mbH$ denotes the set of Hermitian operators. 
The above dual form is known to provide operational meanings to robustness measures, and it particularly indicates that the robustness is physically observable~\cite{Takagi2019general,Yuan2019memory}. 
Furthermore, these two expressions provide useful bounds for the optimal resource cost as
\bal
 2\Tr[YJ_\mU]-1 \leq \gamma_{\rm opt}(\mU) \leq 2s+1
 \label{eq:gamma optimal bound}
\eal
where $Y\in\mbH$ and $s\geq 0$ are any Hermitian operator and real number satisfying the condition in \eqref{eq:robustness primal} and \eqref{eq:robustness dual}.
In addition, we show in Appendix~\ref{app:systematic lower bound} that if we find a decomposition $\mU = \sum_\alpha \eta_\alpha \mO_\alpha=\sum_\alpha \eta_\alpha \mE\circ\mV_\alpha$, which gives an upper bound in \eqref{eq:gamma optimal bound}, then $Y=d^{-2}J_{{\mE^{-1}}^\dagger\circ\mU}$ with $\mE^{-1}:=\sum_\alpha\eta_\alpha\mV_\alpha\circ\mU^\dagger$ satisfies the condition in \eqref{eq:robustness dual} and provides $2\Tr[\Phi_d\id\otimes\mE^{-1}(\Phi_d)]-1$ as a candidate for a good lower bound.  

\section{Evaluation of resource cost}

The generality of our method allows us to obtain bounds for the optimal resource cost for general noise channels. 
To see this, let us consider the noise channels of the form $\mE = (1-\epsilon)\id + \epsilon_+\Lambda - \epsilon_-\Xi$ where $\epsilon, \epsilon_\pm\geq 0$, $\Lambda,\Xi\in\tilde\mP$, which represents a large class of channels and indeed any CPTP map on up to two-qubit systems. 
In Appendix~\ref{app:basis cptp}, we show the latter claim by introducing a universal set of basis operations for CPTP maps that solely consist of programmable operations in $\tilde\mP$, which can be of independent interest.
Although we are usually interested in the cases of small error $\epsilon$, here we do not impose this assumption.
This takes into account the fact that the decomposition is not unique, and the following theorem provides bounds for the optimal cost for any decomposition of this form. 
To state the result, let $I_{ij}=\lset\vec k\in\{0,1\}^i\sbar {\rm wt}(\vec k)=j\rset$ denote the set of $i$-bit strings which have $j$ 1's.
Then, for $\vec k\in I_{ij}$ we define $(A^j,B^{i-j})_{\vec k}$ to be an operation that applies $A$ for $j$ times and $B$ for $i-j$ times with a pattern specified by a binary string $\vec k$.
For instance, if $\vec k=(0, 1, 1)$, then $(A^2,B^1)_{\vec k}=B\circ A\circ A$.
Then, we get the following bounds that provide a systematic estimation of the optimal cost (the proof is given in Appendix~\ref{app:gamma opt general}):

\begin{thm} \label{thm:gamma opt general}
For $\mE = (1-\epsilon)\id + \epsilon_+\Lambda - \epsilon_-\Xi$ where $\mE\in\mT(d)$, $\epsilon, \epsilon_\pm\geq 0$, and $\Lambda,\Xi\in\tilde\mP(d)$,
if $1-\epsilon>\epsilon_++\epsilon_-$, then for any $\mU\in\mT_u(d)$,
\bal
 2\sum_{i=0}^\infty \sum_{j=0}^i \frac{t_{ij}(-\epsilon_+)^j\epsilon_-^{i-j}}{(1-\epsilon)^{i+1}}-1
 \leq \gamma_{\rm opt}(\mU) \leq \frac{1}{1-2\epsilon_+}
 \label{eq:gamma opt bound general}
\eal
where $t_{ij}:=\sum_{\vec k\in I_{ij}}\Tr[\Phi_d\id\otimes ({\Lambda}^{j}, {\Xi}^{i-j})_{\vec k}(\Phi_d)]$.
\end{thm}

The upper bound, an achievable cost (see also Appendix~\ref{app:efficient sampling}), becomes especially insightful when the given error channel is not expressed by Pauli or Clifford operations.
For instance, consider the generalized dephasing channel $\mF_{\hat n, \epsilon}(\rho) := (1-\epsilon) \rho + \epsilon e^{i\hat n \cdot \hat\sigma \pi/2 } \rho e^{-i\hat n \cdot \hat\sigma\pi/2 }$, where $\hat n\cdot \hat\sigma = n_x X + n_y Y + n_z Z$ and $\hat n$ is the unit vector that determines the rotation axis. 
As we show in Appendix~\ref{app:prob flip}, the best systematic decomposition with the Clifford-based basis in Ref.~\cite{Endo2018practical} gives $\gamma = \{1+(\sqrt{2}-1)\epsilon\}/(1-2\epsilon)$ for $\hat n = (\cos(\pi/8), 0, \sin(\pi/8))$, leading to a larger cost than what is achievable in our framework.
This advantage comes from our flexible choice of decomposition basis taking advantage of our large set of implementable operations.
In fact, we find that the advantage is generic.
We compare two costs $\gamma_{\rm disc}$ and $\gamma_{\rm cont}$ where $\gamma_{\rm disc}$ is the optimal cost for the discrete basis in Ref.~\cite{Endo2018practical} and $\gamma_{\rm cont}$ is the achievable cost $(1-2\epsilon_+)^{-1}$ from Theorem~\ref{thm:gamma opt general} realized by our continuous basis, where we set $\mU=\id$ for both cases. 
Figure~\ref{fig:random noise} plots $\gamma_{\rm disc}/\gamma_{\rm cont}$ for randomly sampled single-qubit and two-qubit noise models of the form $\mE=(1-\epsilon)\id + \epsilon \mV$, where $\mV$ is a Haar random unitary. 
We find that $\gamma_{\rm cont}$ is smaller than $\gamma_{\rm disc}$ in many cases for single-qubit error channels and all the sampled cases for two-qubit error channels.
The advantage for two-qubit error is particularly significant---the improvement can become around the factor of 2 for the realistic noise strength, which can drastically reduce the total sampling cost that grows exponentially with the number of gates.  
Since two-qubit noise will be the most demanding one in real experiments, our result may greatly ease the experimental challenges. 
At the same time, the advantage seen in Fig.~\ref{fig:random noise} confirms the necessity of considering the extended class of implementable operations introduced here to properly assess the ultimate capability of error mitigation. 

\begin{figure}
    \centering
    \subfloat[][Single-qubit noise]{
    \includegraphics[width=.24\textwidth]{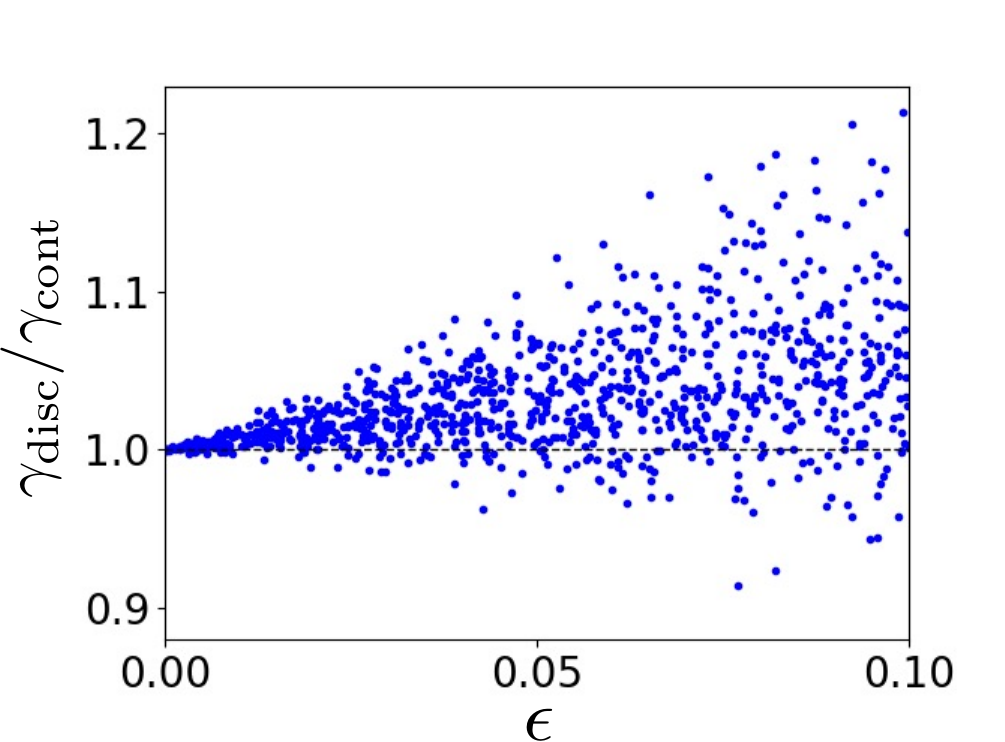}
    }
    \subfloat[][Two-qubit noise]{
    \includegraphics[width=.24\textwidth]{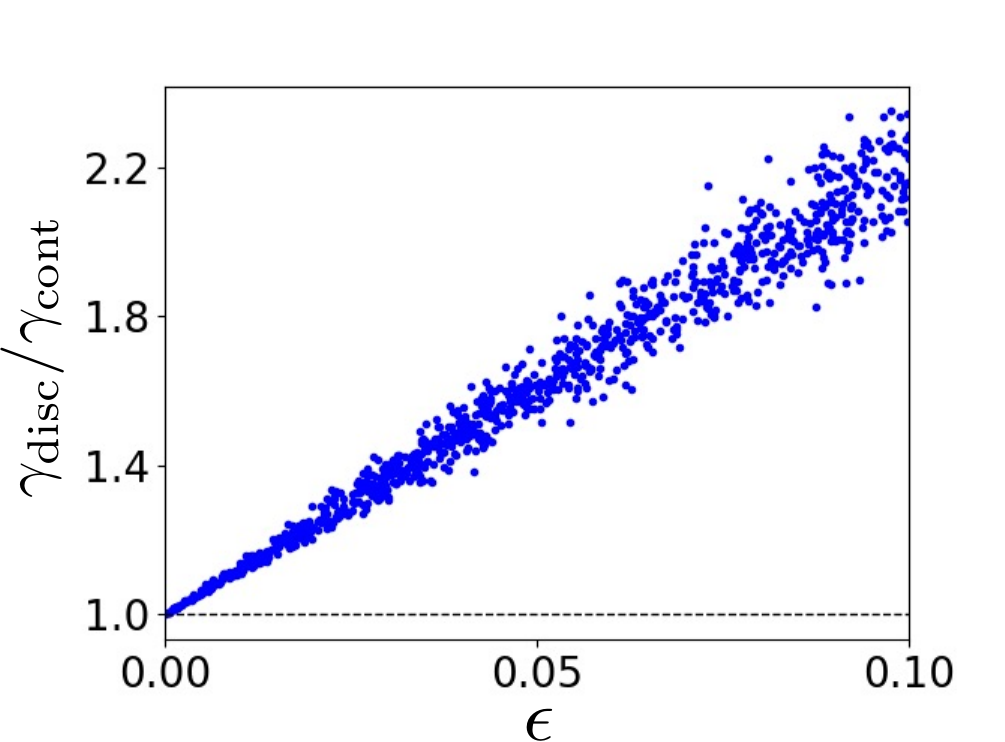}
    }
    \caption{The ratio of the resource cost obtained by the systematic decomposition with discrete basis operations in Ref.~\cite{Endo2018practical} to the achievable cost obtained in Theorem~\ref{thm:gamma opt general} with respect to the noise strength, where the error channel is chosen as  $(1-\epsilon)\id + \epsilon \mV$ with $\mV$ being a Haar random unitary acting on (a) single-qubit systems and (b) two-qubit systems. The black dotted line corresponds to $\gamma_{\rm disc}/\gamma_{\rm cont}=1$. The number of samples used for each figure is $10^3$.}
    \label{fig:random noise}
\end{figure}

The lower bound, on the other hand, corresponds to the fundamental limitations on the error mitigation performance.
It serves as a universal necessary cost even if we employ the continuous set of implementable operations. 
Notably, this bound readily applies to any subset of implementable operations contained in $\mI_\mE$, encompassing a broad class of noisy devices of interest.
Although it might look daunting to evaluate the expression for the lower bound, it can be analytically calculated for many cases, as we discuss in Appendix~\ref{app:gamma opt general}.

Theorem~\ref{thm:gamma opt general} can also give insights into correlated noise models. 
If the noise of interest has correlation among a number of qubits, the effective noise channel gets large and a quasiprobability decomposition becomes numerically intractable even if the description of the noise is given and a discrete set of implementable operations is used. 
On the other hand, Theorem~\ref{thm:gamma opt general} immediately gives an effective evaluation of the optimal cost as long as the noise channel is provided in a certain form.


Interestingly, the bounds in \eqref{eq:gamma opt bound general} can be improved by focusing on specific noise models. 
We can even obtain the \emph{exact} optimal costs for depolarizing and dephasing noise, providing a precise characterization of the devices' capability under these noise models (the proof is given in Appendix~\ref{app:gamma opt depolarizing}).
Together with \eqref{eq:gamma and robustness}, it also gives an operational meaning to the robustness measure in terms of quantum error mitigation.

\begin{thm}\label{thm:robustness depolarizing}
 Let $\mD_{d,\epsilon}$ be a depolarizing channel acting on $d$-dimensional systems defined by
 $\mD_{d,\epsilon}(\rho) := (1-\epsilon) \rho + \epsilon \mbI/d$.
Then, for any unitary gate $\mU\in\mT_u(d)$ and $0\leq \epsilon <1$, 
\bal
 \gamma_{\rm opt}(\mU) = \frac{1+(1-2/d^2)\epsilon}{1-\epsilon},
 \label{eq:dep robustness optimal}
\eal
and the minimum in \eqref{eq:gamma opt definition} is achieved at $\eta_0=1+(d^2-1)\epsilon/\{d^2(1-\epsilon)\}$, $\mO_0=\mD_{d,\epsilon}\circ\mU$ and $\eta_i = -\epsilon/\{d^2(1-\epsilon)\}$, $\mO_i=\mD_{d,\epsilon}\circ\mP_i\circ\mU,$ for $i=1,\dots,d^2-1$ where $\mP_i$ is the $i$\,th Pauli channel.

Also, let $\mF_{\epsilon}:= (1-\epsilon) \id + \epsilon \mZ$ be the qubit dephasing channel.
Then, for any unitary gate $\mU\in\mT_u(2)$ and $0\leq \epsilon <1/2$,
\bal
 \gamma_{\rm opt}(\mU) = \frac{1}{1-2\epsilon}
 \label{eq:flip robustness optimal}
\eal
and the minimum in \eqref{eq:gamma opt definition} is achieved at 
$\eta_0=(1-\epsilon)/(1-2\epsilon)$, $\mO_0=\mF_\epsilon\circ\mU$ and $\eta_1=-\epsilon/(1-2\epsilon)$, $\mO_1=\mF_\epsilon\circ\mZ\circ\mU$.
\end{thm}

This result shows that for these noise models, the optimal cost can be achieved by the Pauli-based basis operations~\cite{Temme2017error,Endo2018practical}, and the continuous degrees of freedom do not help reduce the cost.
This particularly implies that the heuristic linear decomposition for the depolarizing noise considered in Ref.~\cite{Temme2017error} remains optimal even in our extended framework, putting fundamental restrictions on the error mitigation feasibility.

Our method also provides a tighter bound for the amplitude damping noise (the proof is given in Appendix~\ref{app:gamma opt amplitude damping}).
\begin{thm} \label{thm:gamma opt amplitude damping}
 Let $\mA_\epsilon$ be the qubit amplitude damping channel with Kraus operators $A_0 = \dm{0}+\sqrt{1-\epsilon}\dm{1}$, $A_1=\sqrt{\epsilon}\ketbra{0}{1}$.
Then, for any unitary channel $\mU\in\mT_u(2)$ and $0\leq \epsilon<1$, 
\bal
 \frac{\sqrt{1-\epsilon}+\epsilon/2}{1-\epsilon}\leq \gamma_{\rm opt}(\mU) \leq  \frac{1+\epsilon}{1-\epsilon}.
\label{eq:robustness ad bound}
\eal
\end{thm}
The upper bound in \eqref{eq:robustness ad bound} can be achieved by a decomposition $\eta_0=(1+\sqrt{1-\epsilon})/\{2(1-\epsilon)\}$, $\mO_0=\mA_\epsilon\circ\mU$, $\eta_1=(1-\sqrt{1-\epsilon})/\{2(1-\epsilon)\}$, $\mO_1=\mA_\epsilon\circ\mZ\circ\mU$, and $\eta_2=-\epsilon/(1-\epsilon)$, $\mO_2=\mA_\epsilon\circ\mP_{\ket{0}}$, and the lower bound is obtained by finding a good witness operator $Y$ in \eqref{eq:robustness dual}, or alternatively directly evaluating the lower bound in \eqref{eq:gamma opt bound general}.


\section{Conclusions}

We investigated the ultimate potential and limitations of near-term noisy devices in terms of quantum error mitigation.
We pointed out the wide programmability equipped with NISQ devices and formalized their implementable operations to properly assess their full capability. 
We provided a general methodology to evaluate the optimal error mitigation cost with the probabilistic error cancellation method by establishing a connection to the robustness measure that naturally emerges as a resource quantifier in our framework.
We applied our method to a general noise model and obtained universal bounds for the optimal mitigation cost, finding that our framework, which takes into account the flexible choice of implementable operations, leads to a generic advantage over the strategy based on discrete basis sets. 
We also obtained the exact optimal costs for depolarizing and dephasing noise, rigorously showing that the cost for depolarizing channels given in Ref.~\cite{Temme2017error} is optimal even in our extended framework, as well as obtained an improved bound for the amplitude damping channel.

Our consideration can be combined with other frameworks such as learning algorithms for unknown error models~\cite{Czarnik2020clifford, Strikis2020learning,Zlokapa2020deep} and analog error mitigation~\cite{Sun2020continuous}, and may be extended to a broad class of problems for which the quasiprobability sampling is used, including the measurement of observables in variational algorithms and classical simulation of noisy quantum circuits.
The generality of our approach also gives us the freedom to consider various sets of implementable operations, allowing one to tailor the analysis in this work to given physical devices. 
Another important future work is to extend our results toward a unified information-theoretic account of error mitigation and error correction, which will further clarify the true capability of noisy quantum devices.

\begin{acknowledgments}

The author is grateful to Tennin Yan, Nobuyuki Yoshioka, and Xiao Yuan for fruitful discussions, and in particular to Suguru Endo, Kosuke Mitarai, and Yuya O. Nakagawa for valuable discussions and useful comments on the manuscript. This work was supported by National Research Foundation (NRF) Singapore, under its NRFF Fellow programme (Award No. NRF-NRFF2016-02), the Singapore Ministry of Education Tier 1 Grant No. 2019-T1-002-015, NSF, ARO, IARPA, AFOSR, the Takenaka Scholarship Foundation, and the Cross-ministerial Strategic Innovation Promotion Program (SIP), “Photonics and Quantum Technology for Society5.0”. Any opinions, findings and conclusions or recommendations expressed in this material are those of the author(s) and do not reflect the views of National Research Foundation, Singapore.

\end{acknowledgments}


\appendix

\section{Optimal cost and robustness (Eq.~\eqref{eq:gamma and robustness})} \label{app:gamma and robustness}

Here, we show Eq.~\eqref{eq:gamma and robustness} for completeness.
Let $\Psi, \Xi\in\mI_\mE$ be the operations that achieve $\left(\mU+R_{\mI_\mE}(\mU)\Xi\right)/(1+R_{\mI_\mE}(\mU))=\Psi$.
Then, one can rewrite $\mU = (1+R_{\mI_\mE}(\mU))\Psi - R_{\mI_\mE}(\mU)\Xi$, which is a valid decomposition of $\mU$. 
Since $\gamma_{\rm opt}$ is obtained by minimizing over all the possible decompositions, we get $\gamma_{\rm opt}(\mU)\leq 1 + R_{\mI_\mE}(\mU)+R_{\mI_\mE}(\mU)=2R_{\mI_\mE}(\mU)+1$.

On the other hand, suppose that a decomposition $\mU=\sum_\alpha \eta_\alpha \mO_\alpha,\,\mO_\alpha\in\mI_\mE$ achieves $\gamma_{\rm opt}=\sum_\alpha |\eta_\alpha|$.
Then, by separately writing the terms with non-negative coefficients and negative coefficients, one can write
\bal
\mU=\sum_{\eta_\alpha\geq 0}\eta_\alpha \mO_\alpha - \sum_{\eta_\alpha<0}|\eta_\alpha| \mO_\alpha = (1+s)\Psi - s\Xi
\eal
where $s=\sum_{\eta_\alpha< 0}|\eta_\alpha|$, and $\Psi = (1+s)^{-1}\sum_{\eta_\alpha\geq 0}\eta_\alpha \mO_\alpha$, $\Xi=s^{-1}\sum_{\eta_\alpha< 0}|\eta_\alpha| \mO_\alpha$ are CPTP maps. 
Note that we also used that $\mU$ and $\mO_\alpha$ are trace preserving and thus $\sum_\alpha \eta_\alpha = 1$;
also, the fact that $\Psi,\Xi\in\mI_\mE$ is ensured by the convexity of $\mI_\mE$.
Thus, we get $(\mU+s\Xi)/(1+s)\in\mI_\mE,\,\Xi\in\mI_\mE$, which leads to $R_{\mI_\mE}(\mU)\leq s = \sum_{\eta_\alpha< 0}|\eta_\alpha| = (\gamma_{\rm opt}-1)/2$.
This shows the other direction of the inequality, $\gamma_{\rm opt}(\mU)\geq 2R_{\mI_\mE}(\mU)+1$.

\section{Dual form of the robustness (Eq.~\eqref{eq:robustness dual})} \label{app:robustness dual}

Here, we obtain the dual form of the robustness in \eqref{eq:robustness dual} for general resource theories of channels (i.e., \textit{any} finite-dimensional convex and closed set of free channels), in which our case is included as a special case. 
The discussion in this section is based on the argument given for another type of robustness measure (known as generalized robustness) for general resource theories of channels~\cite{Takagi2019general}.

Let $\mT(A,B)$ be the set of quantum channels with input subsystem $A$ and output subsystem $B$. 
Given a convex and closed set of channels $\mO_\mF\subseteq \mT(A,B)$, we consider the (standard) robustness measure for channel $\Lambda\in\mT(A,B)$ with respect to $\mO_\mF$ as 
\bal
R_{\mO_{\mF}}(\Lambda) \coloneqq \min\lset r\geq 0 \sbar \frac{\Lambda +  r \, \Theta}{1+r}\in \mO_{\mF}, \;\, \Theta \in \mO_{\mF}\rset,
\label{eq:def robustness channels}
\eal
where $J_\Lambda= \id\otimes \Lambda (d\cdot\Phi_d)$ denotes the Choi matrix for a channel $\Lambda$ and we assume $R_{\mO_\mF}(\Lambda)<\infty$. 
Let $\mbH_{A(B)}$ be the subspace consisting of Hermitian operators defined on subsystems $A (B)$, and $\mO_\mF^J\subseteq \mbH_A\otimes \mbH_B$ be the set of Choi matrices corresponding to channels in $\mO_\mF$.
Then, introducing the variable $\tilde\Xi = \Lambda + r \Theta$, \eqref{eq:def robustness channels} can be rewritten as the following optimization problem:
\ba
&{\text{\rm minimize}}& \ \ r \label{eq:opt_obj channel}\\
&{\text{\rm subject to}}&\ \ J_{\tilde\Xi} - J_{\Lambda} \in \cone\left(\mO_\mF^J\right)    \label{eq:opt_cond1 channel}\\
&& \ \ J_{\tilde\Xi} \in {\rm cone}\left(\mO_\mF^J\right)\\
&& \ \  \Tr_{B}[J_{\tilde\Xi}]=(1+r)\,\mbI_A
\ea
where $\cone(\mS):=\lset \lambda S \sbar \lambda\geq 0, S\in\mS\rset.$
We write the Lagrangian as
\begin{equation}\begin{aligned}
    L(J_{\tilde\Xi}, r ; X, Y, Z) &= r - \Tr[ Z J_{\tilde\Xi} ] - \Tr[ Y (J_{\tilde\Xi} - J_\Lambda) ]\\
    &\quad- \Tr[ X\left( (1+r) \mbI_A - \Tr_{B} J_{\tilde\Xi} \right)]\\
    &= r ( 1 - \Tr[ X] ) - \Tr[ (Z + Y - X \otimes \mbI_{B}) J_{\tilde\Xi} ]\\
    &\quad+ \Tr[ Y J_\Lambda ] - \Tr[ X ]
\end{aligned}\end{equation}
where $Z, Y\in \mO_\mF^J \*$ and $X\in\mbH_A$ with $\mS \*=\lset W \sbar \Tr[WC]\geq 0,\forall C\in\cone(\mS)\rset$ being a dual cone of $\cone(\mS)$.
This gives the dual problem 
\ba
&{\text{\rm maximize}}& \ \  \Tr[Y J_{\Lambda}] -1  \label{eq:dual1 obj channel}\\
&{\text{\rm subject to}}&\ \ Y=-Z+X\otimes\mbI\in \mO_\mF^J \* \label{eq:dual1 cond1 channel}\\
&& \ \ X\in \mbH_A,\,\Tr[X]=1 \label{eq:dual1 cond2 channel}\\
&& \ \ Z\in \mO_\mF^J \*. \label{eq:dual1 cond3 channel}
\ea

It can be confirmed that the Slater's condition~\cite{boyd_2004} holds by taking $X=\mbI_A/d_A$ and $Z=\mbI_{AB}/(2d_A)$ where $d_A$ is the dimension of $A$.
Thus, the strong duality holds and the dual optimum coincides with the primal optimum. 
Using \eqref{eq:dual1 cond1 channel}, \eqref{eq:dual1 cond2 channel}, and \eqref{eq:dual1 cond3 channel}, we get for any $\Xi\in\mO_\mF$ 
\ba
 \Tr[(-Y+X\otimes\mbI)J_\Xi] &=& -\Tr[YJ_\Xi] + \Tr[\Xi(X)]\\
 &=& -\Tr[YJ_\Xi] + 1 \geq 0.
\ea

Since the objective function does not include $X$ or $Z$, we reach the following equivalent formulation;

\begin{align}
&{\text{\rm maximize}} \ \  \Tr[Y J_{\Lambda}] -1  \label{eq:dual2 obj channel}\\
&{\text{\rm subject to}}\ \ \ Y\in \mO_\mF^J \* \iff \Tr[YJ_\Xi]\geq 0,\ \forall \Xi\in \mO_\mF \\
&\hspace{1.5cm}\Tr[YJ_\Xi]\leq 1,\ \forall \Xi\in \mO_\mF, 
\end{align}
which results in \eqref{eq:robustness dual} by taking $\mO_\mF=\mI_\mE$.

\section{Systematic lower bound for Eq.~\eqref{eq:gamma optimal bound}}\label{app:systematic lower bound}

Here, we present a systematic lower bound that can be constructed by a specific decomposition used for the upper bound. 
Suppose $\mU=\sum_\alpha \eta_\alpha \mE\circ\mV_\alpha$ for $\mV_\alpha\in\tilde\mP$ and define $\mE^{-1}:=\sum_\alpha \eta_\alpha \mV_\alpha\circ\mU^\dagger$ so that $\mE\circ\mE^{-1}=\id$.
Then, as long as $\mV_\alpha,\mE\in\mT(d)$, it also holds that $\mE^{-1}\circ\mE=\id$. 
This can be shown by considering a matrix representation of an operation $\Lambda\in\mT(d)$ with a $d^2\times d^2$ matrix $M^\Lambda_{ab,ij}$ so that a matrix element of $\sigma := \Lambda(\rho)$ is written as $\sigma_{ab} = \sum_{ij}M^\Lambda_{ab,ij}\rho_{ij}$.
Then, $\mE\circ\mE^{-1}=\id$ implies $M^\mE M^{\mE^{-1}}=M^{\id}=\mbI$. 
Since the right inverse of a square matrix is also the left inverse, we get $M^{\mE^{-1}}M^\mE=\mbI$, resulting in $\mE^{-1}\circ\mE=\id$.
Using this, we get
\bal
\frac{1}{d^2}\Tr[J_{{\mE^{-1}}^\dagger\circ\mU}J_{\mE\circ\mV}] &= \Tr[\id\otimes\mU(\Phi_d)\,\id\otimes\mE^{-1}\circ\mE\circ\mV(\Phi_d)]\\
&= \Tr[\id\otimes\mU(\Phi_d)\,\id\otimes\mV(\Phi_d)],
\eal
which ensures $0\leq d^{-2}\Tr[J_{{\mE^{-1}}^\dagger\circ\mU}J_{\mE\circ\mV}]\leq 1$ for any $\mV\in\tilde\mP(d)$. 
Since $J_{{\mE^{-1}}^\dagger\circ\mU}$ is also Hermitian, $Y:=d^{-2}J_{{\mE^{-1}}^\dagger\circ\mU}$ is a valid operator satisfying the condition in \eqref{eq:robustness dual}. 
We can further obtain
\bal
 \Tr[YJ_\mU] &= \Tr[\id\otimes\mU(\Phi_d)\,\id\otimes\mE^{-1}\circ\mU(\Phi_d)]\\
 &=\Tr[\mU^T\otimes\id(\Phi_d)\,\mU^T\otimes\mE^{-1}(\Phi_d)]\\
 &=\Tr[\Phi_d\,\id\otimes\mE^{-1}(\Phi_d)].
\eal
where $\mU^T(\cdot)=U^T\cdot (U^T)^\dagger$ is the unitary transposed with respect to the Schmidt basis of $\Phi_d$, and we also used that $\id\otimes A \ket{\Phi_d} =  A^T\otimes\id  \ket{\Phi_d}$ for any matrix $A$.  
Thus, \eqref{eq:gamma optimal bound} implies that $2\Tr[\Phi_d\,\id\otimes\mE^{-1}(\Phi_d)]-1$ serves as a valid lower bound for $\gamma_{\rm opt}(\mU)$.

\section{Basis operations for CPTP maps} \label{app:basis cptp}

Although any noise channel can be represented by a CPTP map that possibly involves some external systems, one could consider CP trace nonincreasing maps as effective error channels acting on systems of interest.
A trace nonincreasing map on $d$-dimensional quantum systems is an element of the $d^4$-dimensional vector space of linear maps, and thus it can be represented as a linear combination of $d^4$ linearly independent trace nonincreasing maps. 
A specific set of such maps $\{\mB_i\}_{i=1}^{16}$ (Table~\ref{tab:universal endo}) that serves as a universal basis that can decompose any trace nonincreasing map was introduced in Ref.~\cite{Endo2018practical}.
It is a set of 16 linearly independent maps, consisting of 10 Clifford unitaries and 6 trace nonincreasing projections.
Then, any operation acting on qubit systems can be linearly decomposed with respect to this basis, and its tensor product also serves as a universal basis for multiqubit systems. 
They also showed that this basis remains linearly independent after suffering from a noise channel as long as the noise strength is sufficiently small and used this for the quasiprobability decomposition for probabilistic error cancellation. 

Although a universal basis should contain trace nonincreasing maps if the noise model can also be trace nonincreasing in general, in many cases noise models of interest are described as trace-preserving maps.
Then, intuitively, one should be able to find a universal basis only consisting of trace-preserving maps for decomposing an arbitrary CPTP map, which is more favorable for probabilistic error cancellation because using trace nonincreasing maps usually incurs a larger overhead constant due to the larger coefficients in the decomposition needed to complement the non-unit trace (see also the example at the end of this section).
In addition, using many projective measurements is not preferred due to the relatively large measurement error rate. 

Here, we argue that if we restrict our attention to trace-preserving maps, the number of elements necessary for a universal basis is reduced to $d^4-d^2+1$, which themselves are trace preserving.
We then provide a specific set of CPTP maps consisting of unitary and state preparations that serve as a universal basis for CPTP maps on single-qubit and two-qubit systems.
To this end, let $\mO_{\rm CPTP}^J(d)$ be the set of Choi matrices for CPTP maps with input and output being $d$-dimensional quantum systems, and 
\bal
\mV(d):={\rm span}(\mO_{\rm CPTP}^J(d))=\lset\sum_i c_i J_{\mE_i}\sbar c_i\in\mbR,\  \mE_i\in \mT(d)\rset
\eal
be the span of the set of Choi matrices. 
We can show that $\mV(d)$ is equivalent to the subspace
\bal
\tilde\mV(d)\coloneqq\lset X\in \mbH(A,B) \sbar \Tr_B X \propto \mbI_A\rset
\eal
where $\mbH(A,B)$ is the set of Hermitian operators acting on the composite system $AB$, each of which has local dimension $d$.
This can be seen as follows. Since $\Tr_B J_{\mE}=\mbI_A,\forall \mE\in\mT(d)$, it is clear that $\mV(d)\subseteq\tilde\mV(d)$. 
On the other hand, take any $X\in\tilde\mV(d)$.
Then, $X$ can always be decomposed as $X=X_+-X_-$, where $X_{\pm}\geq 0$, and $\Tr_B X_+-\Tr_B X_- = c\mbI_A$ for some constant $c\in\mbR$.
Define 
\bal
\tilde X_+(\lambda)&\coloneqq X_+ +(\lambda\mbI_A-\Tr_B X_+)\otimes\mbI_B/d\\
\tilde X_-(\lambda)&\coloneqq X_- +[(\lambda-c)\mbI_A-\Tr_B X_-]\otimes\mbI_B/d\\
&=X_- +(\lambda\mbI_A-\Tr_B X_+)\otimes\mbI_B/d.
\eal
One can check that $X=X_+ - X_-=\tilde X_+(\lambda) - \tilde X_-(\lambda)$ and $\Tr_B\tilde X_+(\lambda)=\lambda\mbI_A$, $\Tr_B\tilde X_+(\lambda)=(\lambda-c)\mbI_A$ for any $\lambda$.  
Moreover, by taking large enough $\lambda$, we can always ensure that $\tilde X_{\pm}(\lambda)\geq 0$.
This means that there always exist two CPTP maps $\mE_1$, $\mE_2$ and real numbers $c_1$, $c_2$ such that $X=c_1J_{\mE_1}+c_2J_{\mE_2}$. This shows $\tilde\mV(d)\subseteq\mV(d)$, resulting in $\mV(d)=\tilde\mV(d)$.  
$\tilde\mV(d)$ is a subspace of $\mbH(A,B)$ with $\dim \mbH(A,B)=d^4$ and because of the constraints $\tilde\mV(d)$ has, the dimension of $\tilde\mV(d)$ is reduced to $d^4-d^2+1$.
Since $\mV(d)=\tilde\mV(d)$, we have $\dim\mV(d)= \dim\tilde\mV(d) = d^4-d^2+1$.
Thus, $d^4-d^2+1$ linearly independent CPTP maps are necessary and sufficient to linearly decompose an arbitrary CPTP map.

Up to two-qubit systems, such sets of CPTP maps can be explicitly constructed.
For single-qubit systems, we replace $\mB_{11}, \mB_{12}, \mB_{13}$ with trace preserving state preparation channels and remove $\mB_{14}, \mB_{15}, \mB_{16}$, leading to a set of linearly independent CPTP maps $\{\tilde\mB_i\}_{i=1}^{13}$ (Table~\ref{tab:universal cptp}). 
For multiqubit systems, simply tensoring $\{\tilde\mB_i\}_{i=1}^{13}$ is not sufficient to construct a universal basis in general because $13^k<16^k-4^k+1$ for $k\geq 2$.
For two-qubit systems, one needs to find $16^2-4^2+1-13^2 = 72$ additional independent operations to construct a complete basis, and they can be explicitly found as in Table~\ref{tab:universal cptp two qubits}. 
An explicit construction beyond two-qubit systems is left for future work.

Note that since all the operations in Table~\ref{tab:universal cptp},~\ref{tab:universal cptp two qubits} are unitary or state preparations, they are elements of programmable operations $\tilde\mP$ in our framework. 
Thus, for any CPTP map $\mE$ on single-qubit and two-qubit systems, there always exists some decomposition $\mE = (1-\epsilon)\id + \epsilon_+ \Lambda - \epsilon_- \Xi,\ \Lambda,\Xi\in\tilde\mP$ up to unitary.

Using $\{\tilde \mB_i\}$ instead of $\{\mB_i\}$ can save cost for error mitigation. For instance, one can check that the optimal cost $\gamma = \sum_{\alpha} |\eta_{\alpha}|$ to mitigate the amplitude damping channel with $\{\mB_i\}$ results in $(1+2\epsilon)/(1-\epsilon)$ whereas the optimal cost with $\{\tilde\mB_i\}$ is $(1+\epsilon)/(1-\epsilon)$ as in \eqref{eq:identity decomposition ad}.

\begin{table}[htbp]
\begin{minipage}[t]{.45\textwidth}
    \centering
    \begin{tabular}{c|c}
    $\mB_1$ & $\id$ \\
    $\mB_2 $ & $\mX$ \\
    $\mB_3 $ & $\mY$ \\
    $\mB_4 $ & $ \mZ$ \\
    $\mB_5 $ & $ \mK^\dagger\circ\mS^\dagger\circ\mK$ \\
    $\mB_6 $ & $ \mK\circ \mS^\dagger \circ\mK^\dagger$ \\
    $\mB_7 $ & $ \mS^\dagger$ \\
    $\mB_8 $ & $ \mK\mH\mK^\dagger$ \\
    $\mB_9 $ & $ \mH$ \\
    $\mB_{10} $ & $ \mK^\dagger\mH\mK$ \\
    $\mB_{11} $ & $ \mK^\dagger\pi_z \mK$ \\
    $\mB_{12} $ & $ \mK\pi_z \mK^\dagger$ \\ 
    $\mB_{13} $ & $ \pi_z$ \\
    $\mB_{14} $ & $ \mK^\dagger\circ\pi_z \circ \mX\circ\mK$\\
    $\mB_{15} $ & $ \mK\circ\pi_z \circ \mX\circ\mK^\dagger$\\
    $\mB_{16} $ & $ \pi_z \circ \mX$ 
    \end{tabular}
    \caption{Universal basis presented in Ref.~\cite{Endo2018practical} based on Clifford unitaries and projections. Curly letters represent that they are considered as channels, e.g. $\mH(\cdot)=H\cdot H$. Note that $H$, $S$ are the Hadamard gate and the phase gate, and $K=SH$ is the Clifford gate that cycles Pauli operators as $K^\dagger X K = Y$, $K^\dagger Y K = Z$, $K^\dagger Z K = X$. $\pi_z(\cdot)=\left(\frac{\mbI+Z}{2}\right)\cdot\left(\frac{\mbI+Z}{2}\right)$ is a trace nonincreasing projection onto the $\ket{0}$ state.}
    \label{tab:universal endo}
\end{minipage}
\hfill
\begin{minipage}[t]{.45\textwidth}
    \centering
    \begin{tabular}{c|c}
    $\tilde\mB_1$ & $\id=\mB_1$ \\
    $\tilde\mB_2 $ & $\mX=\mB_2$ \\
    $\tilde\mB_3 $ & $\mY=\mB_3$ \\
    $\tilde\mB_4 $ & $ \mZ=\mB_4$ \\
    $\tilde\mB_5 $ & $ \mK^\dagger\circ\mS^\dagger\circ\mK=\mB_5$ \\
    $\tilde\mB_6 $ & $ \mK\circ \mS^\dagger \circ\mK^\dagger=\mB_6$ \\
    $\tilde\mB_7 $ & $ \mS^\dagger=\mB_7$ \\
    $\tilde\mB_8 $ & $ \mK\mH\mK^\dagger=\mB_8$ \\
    $\tilde\mB_9 $ & $ \mH=\mB_9$ \\
    $\tilde\mB_{10} $ & $ \mK^\dagger\mH\mK=\mB_{10}$ \\
    $\tilde\mB_{11} $ & $ \mP_{\ket{+}}$ \\
    $\tilde\mB_{12} $ & $ \mP_{\ket{+y}}$ \\
    $\tilde\mB_{13} $ & $ \mP_{\ket{0}}$ \\
    \end{tabular}
    \caption{Universal basis for CPTP maps acting on single-qubit systems. $\mP_{\ket{\psi}}$ is a channel that prepares a state $\ket{\psi}$. $\ket{+}=\frac{1}{\sqrt{2}}(\ket{0}+\ket{1})$ and $\ket{+y}=\frac{1}{\sqrt{2}}(\ket{0}+i\ket{1})$ are +1 eigenstates of $X$ and $Y$.}
    \label{tab:universal cptp}
\end{minipage}
\end{table}

\begin{table}[htbp]
\centering
\begin{tabular}{c|c}
$\tilde\mB_{1}$--$\tilde\mB_{169}$ & $\{\mB_i\}_{i=1}^{13}\otimes \{\mB_i\}_{i=1}^{13}$\\
$\tilde\mB_{170}$--$\tilde\mB_{178}$ & $\mathcal{CX}$ + conjugation with $\mK_{1,2},\mK_{1,2}^\dagger$ \\
$\tilde\mB_{179}$--$\tilde\mB_{187}$ & $\mX_1\circ\mathcal{CX}\circ\mX_1$ + conjugation with $\mK_{1,2},\mK_{1,2}^\dagger$ \\
$\tilde\mB_{188}$--$\tilde\mB_{196}$ & $\mathcal{CS}$ + conjugation with $\mK_{1,2},\mK_{1,2}^\dagger$ \\
$\tilde\mB_{197}$--$\tilde\mB_{205}$ & $\mathcal{CH}$ + conjugation with $\mK_{1,2},\mK_{1,2}^\dagger$ \\
$\tilde\mB_{206}$--$\tilde\mB_{214}$ & $\mathcal{C_H X}$ + conjugation with $\mK_{1,2},\mK_{1,2}^\dagger$ \\
$\tilde\mB_{215}$--$\tilde\mB_{223}$ & $\mathcal{C X}\circ\mH_1$ + conjugation with $\mK_{1,2},\mK_{1,2}^\dagger$ \\
$\tilde\mB_{224}$--$\tilde\mB_{226}$ & $\mathcal{S}_W$ + conjugation with $\mK_2,\mK_2^\dagger$ \\
$\tilde\mB_{227}$--$\tilde\mB_{232}$ & $i\mathcal{S}_W$ + conjugation with $\mK_{1,2}, \mK_{2}^\dagger$ \\
$\tilde\mB_{233}$--$\tilde\mB_{241}$ & $\mathcal{S}_W\circ\mH_1$ + conjugation with $\mK_{1,2},\mK_{1,2}^\dagger$ 
\end{tabular}
\caption{Universal basis for CPTP maps acting on two-qubit systems. $\mathcal{CX}$, $\mathcal{CS}$, $\mathcal{CH}$, $\mathcal{C_H X}$ are channel versions of CNOT, controlled-phase, controlled-Hadamard, and NOT controlled with $\pm 1$ eigenstates of the Hadamard gate. $\mS_W$ and $i\mS_W$ are channel versions of SWAP and iSWAP ($=\dm{00}+i\ketbra{10}{01}+i\ketbra{01}{10}+\dm{11}$) gates. The subscripts refer to the subsystems that the operations act on. ``$\mU$ + conjugation with $\mV$'' refers to sandwiching $\mU$ with $\mV$ and its conjugation $\mV^\dagger$ as $\mV^\dagger \circ \mU \circ \mV$. Then, ``$\mU$ + conjugation with $\mK_{1,2}, \mK_{1,2}^\dagger$'' collects all of the nine possible conjugations with 
$\id_{12}$, $\mK_1$, $\mK_2$, $\mK_1^\dagger$, $\mK_2^\dagger$, $\mK_1\otimes\mK_2$, $\mK_1\otimes\mK_2^\dagger$, $\mK_1^\dagger\otimes\mK_2$,  $\mK_1^\dagger\otimes \mK_2^\dagger$ for $\mU$.}
\label{tab:universal cptp two qubits}
\end{table}

\section{Proof of Theorem~\ref{thm:gamma opt general}} \label{app:gamma opt general}

\begin{proof}
We first obtain an upper bound. Let us define 
\bal
 \Psi_e = \frac{1}{N_e}\sum_{i=0}^\infty\sum_{j:\text{even}}\sum_{\vec k\in I_{ij}}\eta_{ij\vec k} (\Lambda^j,\Xi^{i-j})_{\vec k}
\eal
\bal
 \Psi_o = \frac{1}{N_o}\sum_{i=1}^\infty\sum_{j:\text{odd}}\sum_{\vec k\in I_{ij}}|\eta_{ij\vec k}| (\Lambda^j,\Xi^{i-j})_{\vec k}
\eal
where
\bal
 \eta_{ij \vec k} = \frac{(-\epsilon_+)^j\epsilon_-^{i-j} }{(1-\epsilon)^{i+1}},\ \forall \vec k\in I_{ij}, \ j=0,\dots,i.
\eal
and $N_{e,o}$ are normalization constants,
\bal
 N_e &= \sum_{i=0}^\infty\sum_{j:\text{even}}\sum_{\vec k\in I_{ij}}\eta_{ij\vec k}\\
 &= \sum_{i=0}^\infty\sum_{j:\text{even}}\binom{i}{j}\frac{\epsilon_+^j\epsilon_-^{i-j}}{(1-\epsilon)^{i+1}}\\
 & = \sum_{i=0}^\infty \frac{1}{2} \frac{(\epsilon_++\epsilon_-)^i + (-\epsilon_++\epsilon_-)^i}{(1-\epsilon)^{i+1}}\\
 &=  \frac{1}{2(1-\epsilon-(\epsilon_++\epsilon_-))} + \frac{1}{2(1-\epsilon-(-\epsilon_++\epsilon_-))}  
\eal
\bal
 N_o &= \sum_{i=1}^\infty\sum_{j:\text{odd}}\sum_{\vec k\in I_{ij}}|\eta_{ij\vec k}| \\
 &= \sum_{i=1}^\infty\sum_{j:\text{odd}}\binom{i}{j}\frac{\epsilon_+^j\epsilon_-^{i-j}}{(1-\epsilon)^{i+1}} \\
 &= \sum_{i=1}^\infty \frac{1}{2} \frac{(\epsilon_++\epsilon_-)^i - (-\epsilon_++\epsilon_-)^i}{(1-\epsilon)^{i+1}}\\
 &=  \frac{1}{2(1-\epsilon-(\epsilon_++\epsilon_-))} - \frac{1}{2(1-\epsilon-(-\epsilon_++\epsilon_-))}
\eal
where we changed the order of the infinite sum because $\sum_{i=0}^\infty \left|\frac{(\epsilon_++\epsilon_-)^i}{(1-\epsilon)^{i+1}}\right|<\infty$ and $\sum_{i=0}^\infty \left|\frac{(-\epsilon_++\epsilon_-)^i}{(1-\epsilon)^{i+1}}\right|<\infty$.
Noting that $\tilde\mP$ is closed under concatenation, i.e., $\mV_1,\mV_2\in\tilde\mP\Rightarrow \mV_1\circ\mV_2\in\tilde\mP$, as well as the convexity of $\tilde\mP$, we have $\Psi_e,\Psi_o\in\tilde\mP$.
Then,
\begin{align}\begin{aligned}
 &\mE\circ \left[N_e \Psi_e - N_o \Psi_o\right] \\
 &\quad\quad=  \left[(1-\epsilon)\id + \epsilon_+\Lambda-\epsilon_-\Xi\right]\circ\\
 &\quad\quad\quad\lim_{n\to\infty}\sum_{i=0}^n\sum_{j=0}^i\sum_{\vec k\in I_{ij}}\left(\frac{(-\epsilon_+)^j\epsilon_-^{i-j}}{(1-\epsilon)^{i+1}}  (\Lambda^j,\Xi^{i-j})_{\vec k}\right)\\
 &\quad\quad= \id + \lim_{n\to\infty}\sum_{j=0}^{n+1}\sum_{\vec k \in I_{n+1j}}\frac{(-\epsilon_+)^j\epsilon_-^{n+1-j}}{(1-\epsilon)^{n+1}}  (\Lambda^j,\Xi^{n+1-j})_{\vec k}\\
 &\quad\quad=\id
\end{aligned}\end{align}
where the last equality is because 
\begin{align}\begin{aligned}
&\left\| \sum_{j=0}^{n+1}\sum_{\vec k\in I_{n+1 j}}\frac{(-\epsilon_+)^j\epsilon_-^{n+1-j}}{(1-\epsilon)^{n+1}}  (\Lambda^j,\Xi^{i-j})_{\vec k}\right\|_\diamond \\
&\quad\leq d^2\frac{\sum_{j=0}^{n+1}\binom{n+1}{j} \epsilon_+^j\epsilon_-^{n+1-j}}{\left(1-\epsilon\right)^{n+1}} \\
 &\quad= d^2\left(\frac{\epsilon_++\epsilon_-}{1-\epsilon}\right)^{n+1}\to 0\ (n\to \infty),
\end{aligned}\end{align}
where $d$ is the dimension of the system. 
Thus, we have $\mU=N_e \mE\circ\Psi_e\circ\mU - N_o\mE\circ\Psi_o\circ\mU$, and since $\Psi_e\circ\mU, \Psi_o\circ\mU\in\tilde\mP$, we get $\gamma_{\rm opt}(\mU)\leq N_e + N_o = \frac{1}{1-\epsilon-(\epsilon_++\epsilon_-)} = \frac{1}{1-2\epsilon_+}$.

Next, we obtain the lower bound using the dual form of the robustness \eqref{eq:robustness dual}. 
Take the following witness 
\bal
 Y &= \frac{1}{d^2}J_{{\mE^{-1}}^\dagger\circ\mU} \\
 &=  \frac{1}{d}\sum_{i=0}^\infty\sum_{j=0}^{i}\sum_{\vec k\in I_{ij}}\frac{(-\epsilon_+)^{j}\epsilon_-^{i-j}}{(1-\epsilon)^{i+1}}\id\otimes\left[({\Lambda^\dagger}^{j},{\Xi^\dagger}^{i-j})_{\vec k}\circ\mU\right](\Phi_d).
\eal

One can check that this is a well-defined bounded operator by
\bal
 \|Y\|_\infty &\leq  \frac{1}{d}\sum_{i=0}^\infty\sum_{j=0}^{i}\sum_{\vec k\in I_{ij}}\frac{\epsilon_+^{j}\epsilon_-^{i-j}}{(1-\epsilon)^{i+1}}\left\|\id\otimes ({\Lambda^\dagger}^{j}, {\Xi^\dagger}^{i-j})_{\vec k}\circ\mU(\Phi_d)\right\|_\infty \\
 &\leq  \frac{1}{d}\sum_{i=0}^\infty\sum_{j=0}^{i}\sum_{\vec k\in I_{ij}}\frac{\epsilon_+^{j}\epsilon_-^{i-j}}{(1-\epsilon)^{i+1}}\max_{i,j,\vec k}|||({\Lambda^\dagger}^{j}, {\Xi^\dagger}^{i-j})_{\vec k}|||_\infty\\
 & = \frac{1}{d}\frac{1}{1-\epsilon-(\epsilon_++\epsilon_-)}\max_{i,j,\vec k}\|(\Lambda^{j},\Xi^{i-j})_{\vec k}\|_\diamond\\
 & = \frac{1}{d}\frac{1}{1-\epsilon-(\epsilon_++\epsilon_-)}
\eal
where, in the first inequality, we used the triangle inequality of the operator norm; in the second inequality, we used the property of the completely bounded infinite norm $\|(\id\otimes\cdot )(X)\|_\infty\leq ||| \cdot |||_\infty:=\sup_X\|(\id\otimes\cdot )(X)\|_\infty$; and in the first equality, we used the dual property between the completely bounded infinite norm and diamond norm~\cite{Watrous2009SDP}. 
Thus, $Y$ is a valid bounded Hermitian operator. 

As argued in Appendix~\ref{app:systematic lower bound}, this choice of $Y$ satisfies the condition in \eqref{eq:robustness dual}, and defining 
\bal
t_{ij}:=\sum_{\vec k\in I_{ij}}\Tr[\Phi_d\id\otimes ({\Lambda}^{j}, {\Xi}^{i-j})_{\vec k}(\Phi_d)],
\eal
we get
\bal
 \Tr[YJ_\mU] &= \Tr[\Phi_d \id\otimes \mE^{-1}(\Phi_d)] = \sum_{i=0}^\infty\sum_{j=0}^i\frac{t_{ij}(-\epsilon_+)^j\epsilon_-^{i-j}}{(1-\epsilon)^{i+1}}.
\eal
Using \eqref{eq:gamma optimal bound} results in the lower bound of the statement.  
\end{proof}

Although it might look daunting to evaluate $t_{ij}$, it can actually be exactly obtained for many cases. 
An important observation is that if at least one state preparation is involved, $t_{ij}$ collapses to a constant. 
For instance, for the amplitude damping channel $\mA_{\delta}$ where $\epsilon = \frac{1+\delta-\sqrt{1-\delta}}{2}$, $\epsilon_+ = \delta$, $\epsilon_-=\frac{\sqrt{1-\delta}-(1-\delta)}{2}$ and $\Lambda=\mP_{\ket{0}}$, $\Xi=\mZ$, we get 
\bal 
t_{ij} =
\begin{cases}
 \frac{1}{4}\binom{i}{j} & (j\neq 0) \\
 1 & (j=0,\ i:{\rm even}) \\
 0 & (j=0,\ i:{\rm odd}),
\end{cases}
\eal
which gives
\begin{align}\begin{aligned}
&\sum_{i=0}^\infty\sum_{j=0}^i\frac{t_{ij}(-\epsilon_+)^j\epsilon_-^{i-j}}{(1-\epsilon)^{i+1}}\\
&\quad= \sum_{i=0}^\infty\sum_{j=1}^i\frac{1}{4}\binom{i}{j}\frac{(-\epsilon_+)^j\epsilon_-^{i-j}}{(1-\epsilon)^{i+1}} + \sum_{i:{\rm even}}\frac{\epsilon_-^{i}}{(1-\epsilon)^{i+1}} \\
&\quad = \sum_{i=0}^\infty\frac{1}{4}\left(\frac{(-\epsilon_++\epsilon_-)^{i}}{(1-\epsilon)^{i+1}}-\frac{\epsilon_-^{i}}{(1-\epsilon)^{i+1}}\right) + \sum_{i:{\rm even}}\frac{\epsilon_-^{i}}{(1-\epsilon)^{i+1}} \\
&\quad = \frac{1}{4}\left(1-\frac{1}{1-\epsilon-\epsilon_-}\right) + \frac{1-\epsilon}{(1-\epsilon-\epsilon_-)(1-\epsilon+\epsilon_-)}\\
&\quad = -\frac{1}{4}\frac{\epsilon+\epsilon_-}{1-\epsilon-\epsilon_-} + \frac{1-\epsilon}{(1-\epsilon-\epsilon_-)(1-\epsilon+\epsilon_-)}.
\end{aligned}\end{align}

Substituting $\epsilon = \frac{1+\delta-\sqrt{1-\delta}}{2}$, $\epsilon_+ = \delta$, and $\epsilon_-=\frac{\sqrt{1-\delta}-(1-\delta)}{2}$, we get $\gamma_{\rm opt}(\mU)\geq \frac{\sqrt{1-\delta}+\delta/2}{1-\delta}$, which reproduces the lower bound in \eqref{eq:robustness ad bound}.

\section{Efficient sampling for general noise channels} \label{app:efficient sampling}

As can be seen in the proof of Theorem~\ref{thm:gamma opt general}, the upper bound can be achieved by a decomposition composed of operations applying $\Lambda$ and $\Xi$ in different patterns.
Although it involves an infinite series of operations, an efficient sampling from the probability distribution originating from the infinite series is possible, which allows one to run the probabilistic error cancellation at this cost in practice. 

To run the probabilistic error cancellation at the cost $1/(1-2\epsilon_+)$ for an error channel $\mE=(1-\epsilon)\id + \epsilon_+\Lambda - \epsilon_-\Xi$, one needs to sample each operation to be applied at the right probability. 
Specifically, the proof of Theorem~\ref{thm:gamma opt general} implies that the operation $(\Lambda^j,\Xi^{i-j})_{\vec k}$ should be applied at probability 
\bal
 P_{ij\vec k}:=\frac{|\eta_{ij\vec k}|}{\sum_{i=0}^\infty \sum_{j=0}^i\sum_{\vec k\in I_{ij}}|\eta_{ij\vec k}|} = \frac{1-\epsilon-(\epsilon_++\epsilon_-)}{(1-\epsilon)^{i+1}}\epsilon_+^j\epsilon_-^{i-j}.
\eal
This sampling can be easily done as follows. First, note that the probability of realizing an element with index $i$ is given by
\bal
 P_i:= \sum_{j=0}^i\sum_{\vec k\in I_{ij}}P_{ij \vec k} = \left(1-\frac{\epsilon_++\epsilon_-}{1-\epsilon}\right)\left(\frac{\epsilon_++\epsilon_-}{1-\epsilon}\right)^i,
\eal
which is the probability of observing heads for $i$ times in sequence followed by a tail at the $i+1$\,th trial when flipping a biased coin with $p_{\rm head}=(\epsilon_++\epsilon_-)/(1-\epsilon)$. 
Also, note that  
\bal
 \frac{P_{ij}}{P_{i}} := \frac{\sum_{\vec k\in I_{ij}}P_{ij \vec k}}{P_{i}} =  \binom{i}{j}\left(\frac{\epsilon_+}{\epsilon_++\epsilon_-}\right)^j\left(\frac{\epsilon_-}{\epsilon_++\epsilon_-}\right)^{i-j},
\eal
which is a binomial distribution and
\bal
\frac{P_{ij \vec k}}{P_{ij}} = \frac{P_{ij \vec k'}}{P_{ij}},\ \forall \vec k, \vec k'\in I_{ij}.
\eal
Thus, the sampling procedure is summarized as follows:
\begin{enumerate}
    \item Flip a biased coin with $p_{\rm head} = (\epsilon_++\epsilon_-)/(1-\epsilon)$ until a tail is observed. Suppose the tail was observed at the $i+1$\,th trial.
    \item Flip another biased coin with $p_{\rm head} = \epsilon_+/(\epsilon_++\epsilon_-)$ for $i$ times.
    Suppose heads were observed for $j$ times.
    \item Randomly choose $\vec k$ out of $\binom{i}{j}$ possible vectors in $I_{ij}$ and apply $(\Lambda^j,\Xi^{i-j})_{\vec k}$.
\end{enumerate}

\section{Cost for the generalized dephasing channel with the universal basis in Ref.~\cite{Endo2018practical}} \label{app:prob flip}

Consider the generalized dephasing channel $\mF_{\hat n, \epsilon}(\rho) := (1-\epsilon) \rho + \epsilon e^{i\hat n \cdot \hat\sigma \pi/2 } \rho e^{-i\hat n \cdot \hat\sigma\pi/2 }$ where $\hat n\cdot \hat\sigma = n_x X + n_y Y + n_z Z$ and $\hat n$ is the unit vector that determines the rotation axis. 
As an example, take $\hat n = (\cos(\pi/8), 0, \sin(\pi/8))$, which is the $\pi$ rotation with respect to the axis halfway between the $X$ rotation and the Hadamard rotation. 

Let us consider an optimal decomposition 
$\id = \sum_i \eta_i \mF_{\hat n , \epsilon}\circ\mB_i$, which is equivalent to decomposing the inverse map $\mF_{\hat n, \epsilon}^{-1}=\sum_i \eta_i \mB_i$,
using the universal basis in Ref.~\cite{Endo2018practical} (see Table~\ref{tab:universal endo} in Appendix~\ref{app:basis cptp}).
One can check that 
\bal
 \id = \mF_{\hat n,\epsilon}\circ\left[\frac{1-\epsilon}{1-2\epsilon}\id + \frac{(\sqrt{2}-1)\epsilon}{2(1-2\epsilon)}\mZ-\frac{\epsilon}{2(1-2\epsilon)}\mX\right.\\
 \left. -\frac{\sqrt{2}\epsilon}{2(1-2\epsilon)}\mH\right],
\eal
with 
$\gamma = \{1+(\sqrt{2}-1)\epsilon\}/(1-2\epsilon)$ is optimal with this basis (as well as the basis in Table~\ref{tab:universal cptp} in Appendix~\ref{app:basis cptp}).
This is greater than the upper bound in \eqref{eq:gamma opt bound general}, and this gap comes from the fact that the upper bound in \eqref{eq:gamma opt bound general} is obtained by adaptively choosing an appropriate basis using the known information about the noise channel $\mF_{\hat n, \epsilon}$.

\section{Proof of Theorem~\ref{thm:robustness depolarizing}} \label{app:gamma opt depolarizing}

\begin{proof}
 We first show the upper bound using \eqref{eq:robustness primal}. 
 To this end, we consider a specific decomposition of the identity map 
 \bal
 \id = \left(1+\frac{(d^2-1)\epsilon}{d^2(1-\epsilon)}\right) \mD_{d,\epsilon}\circ\id - \frac{\epsilon}{d^2 (1-\epsilon)} \sum_{i=1}^{d^2-1} \mD_{d,\epsilon} \circ \mP_i.
 \eal
By applying $\mU$ to both sides from the right and noting that $\mP_i\circ\mU\in\mT_u$, we get $\gamma_{\rm opt}(\mU)\leq \{1+(1-2/d^2)\epsilon\}/(1-\epsilon)$ by \eqref{eq:gamma opt definition}.
On the other hand, consider the following $Y\in\mbH$ providing the systematic lower bound
\bal
Y = d^{-2}J_{{\mD_{d,\epsilon}^{-1}}^\dagger\circ\mU} =\frac{1}{d}\frac{1}{1-\epsilon}\id\otimes\mU(\Phi_d) - \frac{1}{d}\frac{\epsilon}{d^2(1-\epsilon)}\mbI.
\eal
A straightforward computation together with \eqref{eq:gamma optimal bound} leads to
\bal
\gamma_{\rm opt}(\mU)&\geq 2\Tr[\Phi_d\id\otimes\mD_{d,\epsilon}^{-1}(\Phi_d)]-1\\&=\{1+(1-2/d^2)\epsilon\}/(1-\epsilon),
\eal
concluding the proof.

The case for the dephasing noise can be shown similarly. 
We first get an upper bound by providing a specific decomposition. 
Namely, consider the following decomposition
\bal
 \id &= \frac{1-\epsilon}{1-2\epsilon} \mF_\epsilon\circ\id - \frac{\epsilon}{1-2\epsilon}  \mF_\epsilon \circ \mZ, 
 \label{eq:identity decomposition flip}
\eal
which can be explicitly checked as 
\begin{align}\begin{aligned}
 &\left((1-\epsilon)\id + \epsilon \mZ\right)\circ \left(\frac{1-\epsilon}{1-2\epsilon} \id - \frac{\epsilon}{1-2\epsilon} \mZ\right)\\
 &\quad= \left(\frac{(1-\epsilon)^2}{1-2\epsilon}-\frac{\epsilon^2}{1-2\epsilon}\right)\id + \left(-\frac{\epsilon(1-\epsilon)}{1-2\epsilon}+\frac{\epsilon(1-\epsilon)}{1-2\epsilon}\right)\mZ\\
 &\quad = \id.
\end{aligned}\end{align}
By applying $\mU$ to both sides of \eqref{eq:identity decomposition flip} from the right and using \eqref{eq:gamma opt definition}, we get $\gamma_{\rm opt}(\mU)\leq\frac{1}{1-2\epsilon}$.

Next, we obtain a lower bound using \eqref{eq:gamma optimal bound} and the dual form of the robustness \eqref{eq:robustness dual}.
Consider the following $Y\in\mbH$.
\bal
 Y &= \frac{1}{4}J_{{\mF_\epsilon^{-1}}^\dagger\circ\mU}\\
 &=\frac{1}{2}\frac{1-\epsilon}{1-2\epsilon}\id\otimes\mU (\Phi_2) - \frac{1}{2}\frac{\epsilon}{1-2\epsilon}(\id\otimes(\mZ\circ\mU))(\Phi_2).
 \label{eq:witness flip}
\eal 
As argued in Appendix~\ref{app:systematic lower bound}, this choice of $Y$ satisfies the condition in \eqref{eq:robustness dual}.
A lower bound of the robustness can be computed using this $Y$ as
\bal
 \Tr[YJ_\mU]-1 &= \Tr[\Phi_2\,\id\otimes\mF_{\epsilon}^{-1}(\Phi_2)]-1\\
 &=\frac{1-\epsilon}{1-2\epsilon} - \frac{\epsilon}{1-2\epsilon}\Tr[\Phi_2\id\otimes\mZ (\Phi_2)] - 1 \\
 &= \frac{\epsilon}{1-2\epsilon}
\eal

Thus, the dual form of the robustness \eqref{eq:robustness dual} implies $R_{\mI_{\mF_\epsilon}}(\mU)\geq \frac{\epsilon}{1-2\epsilon}$, and using \eqref{eq:gamma optimal bound}, we get $\gamma_{\rm opt}(\mU)\geq \frac{1}{1-2\epsilon}$. Combining it with the matching upper bound shown above concludes the proof of the statement. 
\end{proof}

\section{Proof of Theorem~\ref{thm:gamma opt amplitude damping}} \label{app:gamma opt amplitude damping}

\begin{proof}
Consider the following decomposition
\bal
 \id &= \frac{1+\sqrt{1-\epsilon}}{2(1-\epsilon)}\mA_\epsilon\circ\id+\frac{1-\sqrt{1-\epsilon}}{2(1-\epsilon)} \mA_\epsilon\circ\mZ - \frac{\epsilon}{1-\epsilon}  \mA_\epsilon \circ \mP_{\ket{0}},
 \label{eq:identity decomposition ad}
\eal
which proves $\gamma_{\rm opt}(\mU)\leq\frac{1+\epsilon}{1-\epsilon}$.

To obtain a lower bound using the dual form \eqref{eq:robustness dual}, consider the following operator
\bal
 Y &= \frac{1}{4}\,J_{{\mA_\epsilon^{-1}}^\dagger\circ\mU}\\
 &=\frac{1}{2}\frac{1+\sqrt{1-\epsilon}}{2(1-\epsilon)}\,\id\otimes\mU (\Phi_2)\\
 &\quad+ \frac{1}{2}\frac{1-\sqrt{1-\epsilon}}{2(1-\epsilon)}(\id\otimes(\mZ\circ\mU))(\Phi_2)- \frac{1}{4} \frac{\epsilon}{1-\epsilon}\, \mU^T(\dm{0})\otimes\mbI
 \label{eq:witness ad}
\eal

As argued in Appendix~\ref{app:systematic lower bound}, this choice of $Y$ brings us to
\bal
 0\leq \Tr[YJ_{\mA_\epsilon\circ\mV}] = \Tr[\id\otimes\mU (\Phi_2)\id\otimes\mV (\Phi_2)] \leq 1,\ \forall \mV\in\tilde\mP
\eal
ensuring the condition in \eqref{eq:robustness dual}.
Finally, using the form of \eqref{eq:witness ad}, we get a lower bound of the robustness
\bal
\Tr[YJ_{\mU}]-1 &= \Tr[\Phi_2\,\id\otimes\mA_{\epsilon}^{-1}(\Phi_2)]-1\\
&= \frac{\sqrt{1-\epsilon}-1+3\epsilon/2}{2(1-\epsilon)}.
\eal
Using \eqref{eq:gamma optimal bound}, we obtain $\gamma_{\rm opt}(\mU)\geq \frac{\sqrt{1-\epsilon}+\epsilon/2}{1-\epsilon}$, which concludes the proof. 
\end{proof}

\bibliographystyle{apsrmp4-2}
\bibliography{myref}

\end{document}